\documentclass[letterpaper, 10 pt, conference]{ieeeconf}  

\IEEEoverridecommandlockouts                              
\usepackage[utf8]{inputenc}
\usepackage{csquotes}
\usepackage{amsmath,amsfonts,amssymb}
\usepackage{graphicx}
\usepackage{array,multirow,makecell}
\usepackage[T1]{fontenc} 
\graphicspath{ {./ } }
\usepackage{bbding}   

\usepackage{xcolor}

\makeatletter
\let\NAT@parse\undefined
\makeatother
\usepackage[pdftex,colorlinks=false]{hyperref}
\usepackage[noabbrev]{cleveref}  
\usepackage{stmaryrd}   

\newcommand{\comment}[1]{}

\def\Rset{\mathbb{R}}
\def\Nset{\mathbb{N}}

\def\cN{{\cal N}}

\def\loc{{\mathrm{loc}}}
\def\app{{\mathrm{app}}}
\def\tot{{\mathrm{tot}}}

\def\bfE{{\mathbf{E}}}

\renewcommand{\t}{^{\mbox{\tiny\sf T}}}

\def\crit{{\mathrm {crit}}}
\def\app{{\mathrm{app}}}
\def\tot{{\mathrm{tot}}}

\definecolor{vert}{rgb}{0.06, 0.7, 0.6}   
\definecolor{mauve}{rgb}{0.6, 0.2, 0.99}   
\definecolor{bleu}{rgb}{0, 0.45, 1}   

\newtheorem{theorem}{Theorem}

\newtheorem{remark}{Remark}

\title{\LARGE \bf
Basic offspring number and robust feedback design for the biological control of vectors by sterile insect release technique*}

\author{Pierre-Alexandre Bliman$^{\, 1}$
\thanks{*The financial support of the STIC AmSud program is acknowledged, through the project 23-STIC-02 BIO-CIVIP
`Biological control of insects vectors and insects pests'.}
\thanks{$^{1}$Pierre-Alexandre Bliman is with Sorbonne Université, Université Paris Cité, CNRS, Inria, Laboratoire Jacques-Louis Lions, LJLL, EPC MUSCLEES, F-75005 Paris, France \href{mailto:pierre-alexandre.bliman@inria.fr}{\tt\small pierre-alexandre.bliman@inria.fr}}
}

\begin{document}

\maketitle
\thispagestyle{empty}
\pagestyle{empty}

\begin{abstract}
Sterile Insect Technique (SIT) is a promising control method against insect pests and insect vectors. It consists in releasing males previously sterilized in laboratory, in order to reduce or eliminate a specific wild population.
We study in this paper the implementation by feedback control of SIT-based elimination campaign of {\em Aedes} mosquitoes.
We provide state-feedback and output-feedback control laws and establish their convergence, as well as their robustness properties.
In this design procedure, a pivotal role is played by the
basic offspring number, and by the use of properties of monotone systems.
\end{abstract}


\section{INTRODUCTION}
\label{se0}

Controlling mosquitoes that transmit established or potentially (re)emerging diseases, such as {\em Aedes aegypti} and {\em Aedes albopictus}, vectors of dengue, chikungunya and Zika; and controlling pests that threaten agriculture, such as the Mediterranean fly ({\em Ceratitis capitata}) or the oriental fly ({\em Bactrocera dorsalis}), which oviposit under the skin surface of their host fruits and damage the harvest, is likely to become even more necessary in the future, including in the temperate zones of the world, due to global climate changes.
Control through insecticides is now used with reluctance, due to impact on biodiversity and  resistance-induced reduced efficiency.
Alternative methods include biological control, which uses beneficial insects or pathogens that they transport to control unwanted insects, weeds, or diseases, and usually target specific species, without harming others.

The method we are interested in here is the Sterile Insect Technique~(SIT)~\cite{Hendrichs:2009aa,Alphey:2010aa}, involving release of males previously sterilized (usually by irradiation, or by an alternative technique like the genetic RIDL technique or the use of Wolbachia-infected males).
The females inseminated by these males do not produce viable eggs, so that sustained releases on an area-wide basis may succeed in reducing the wild population.
Designing practically successful release campaigns on a large scale for limited cost is still a source of important questions.
In this paper, we are interested in the implementation of SIT on {\em Aedes} mosquitoes by feedback control.

SIT model and analysis results have been presented in~\cite[(1)]{Esteva:2005aa}~and \cite[(1)]{Thome:2010aa} for {\em Aedes aegypti} mosquitoes, vectors of dengue fever and other arboviruses, allowing the study of epidemic spread~\cite{Dumont:2012aa}.
This species has the particularity of experiencing intraspecific competition, possibly through reduction of the oviposition rate in congested breeding sites.
Continuing in the same vein,~\cite[(2)]{Bossin:2019aa} proposed a simpler, 3-dimensional, model including Allee effect.
As another example, intraspecific competition in the species {\em Anopheles}, vectors of malaria, occurs later in the aquatic stages of life, through mortality increase.
Such a modelling option has been explored in~\cite[(28-31)]{Anguelov:2012aa}, subsequently simplified in~\cite[(16)]{Anguelov:2020aa}, resulting in a nearby 3-dimensional model.
In the sequel we use a 3-dimensional model simply obtained by ignoring the Allee effect in~\cite[(2)]{Bossin:2019aa}.
This model, numbered~\eqref{eq2} below, is similar to some model in~\cite{Almeida:2019aa} but does not make the limiting assumptions of sex ratio 1:1 at birth and of identical male/female mortality rates, which reduce the system dimension.
It has been used for feedback control in~\cite{Bidi:2025aa,Bidi:2023ab,Cristofaro:2024aa} and this allows comparison of the techniques.

Notice that from a control theory perspective, eliminating mosquito population is a problem of stabilization of the extinction equilibrium.
A key idea to achieve this goal in a biological control method such as SIT, is to design control strategies that ensure a sufficient proportion of sterile males within the total male population, in order to reduce the basic offspring number of the population to a subcritical value.
This principle was used in~\cite{Bliman:2019aa} to design successful release campaign by periodic impulsive control, with the size of every release adjusted in accordance with the current estimate of the wild population size.
It was also used in~\cite{Bliman:2021aa} as basis for feedback synthesis in slightly different context (replacement of wild mosquito population by {\em Wolbachia}-infected mosquitoes, whose vectorial capacity is reduced).
These results were obtained in relatively simple situations, and understanding how to apply the technique to more complex models and conditions is a challenge we tackle here. 

In~\cite{Bidi:2025aa}, backstepping control~\cite{Coron:2007aa} was used to build state feedback law stabilizing the extinction equilibrium of system~\eqref{eq2}, denoted $\bfE_0$ in the sequel.
The idea therein was to stabilize the sterile-to-wild male ratio at a given value, chosen sufficiently large in order to force elimination of the mosquito population.
%
%
Static output feedback laws were also proposed in~\cite{Bidi:2025aa}, based on measured output that is either the total number of males, or the number of wild males, but stabilization was not proved in general for the corresponding policies.
Numerical simulations were provided and parametric robustness was tested numerically.
The approach taken in~\cite{Bidi:2023ab} combined reinforcement learning with mathematical analysis to identify a candidate solution for an explicit stabilizing feedback control.
The latter was based on measuring the total number of males and the total number of females in the field.
However convergence of the controlled system was only suggested by numerical simulations.
Last, in~\cite{Cristofaro:2024aa}, backstepping was used to design stabilizing control law for a reduced model and to test it numerically on the complete model~\eqref{eq2}.

We propose in the present paper a feedback control law significantly simpler than~\cite{Bidi:2025aa}, capable of achieving elimination for a population evolving according to system~\eqref{eq2}.
One of our aims is to highlight the role of the basic offspring number in biological control, and the benefits of using the monotone system theory~\cite{Smith:1995aa,Smith:2017aa}.
The model is presented in~\Cref{se1} and the control principle is described in~\Cref{se2}.
State-feedback synthesis is then achieved in~\Cref{se3}.
Taking advantage of the monotonicity properties~\cite{Smith:1995aa,Smith:2017aa} of the life cycle model, we then show in~\Cref{se7} how to ensure robust stabilization against parametric and dynamic uncertainty.
It is then exposed in~\Cref{se4} how interval observers constitute a natural way to synthesize output-based feedback laws when only partial state measurement is available.
Due to lack of space, proofs of the results and numerical illustrations are omitted and will be published in a forthcoming paper.




%


\paragraph*{Notations}
For any $z\in\Rset$, define $|z|_+ := \max\{z;0\}$.
A vector $x \in \Rset^n$, $n\in\Nset$, is said {\em nonnegative} if all its components are nonnegative;~{\em positive} if it is nonnegative and nonzero;~{\em strictly positive} if all its components are positive.
Defining the corresponding order relations, these properties are written respectively $x \geq 0_n$, resp.~$x > 0_n$, resp.~$x \gg 0_n$.
Any two elements $x<y$ generate the {\em order interval}~\cite{Smith:1995aa} $\llbracket x,y \rrbracket := \{ z\ :\ x \leq z \leq y\} $.
A matrix is said to be {\em Metzler} if all its off-diagonal components are nonnegative.
The {\em stability modulus}~\cite[p.~32]{Kaszkurewicz:2012aa} (or {\em spectral abscissa}~\cite{Deutsch:1975aa}) of a matrix, is the greatest real part of its eigenvalues.

\section{MODEL}
\label{se1}

The uncontrolled model that describes the dynamics of the wild mosquito population alone is first provided in~\eqref{eq1}.
Its state variable incorporates three compartments, representing aquatic phase ($E$), males ($M$) and fertilized females ($F$).
The model controlled through sterile male ($M_s$) releases is then presented in~\eqref{eq2}.

\paragraph*{$\bullet$ Uncontrolled system}

\begin{subequations}
\label{eq1}
\begin{eqnarray}
\label{eq1a}
\dot E
& = &
\beta_E F \left(
1 - \frac{E}{K}
\right) - (\nu_E+\delta_E) E,\\
\label{eq1b}
\dot M
& = &
(1-\nu)\nu_E E - \delta_M M,\\
\label{eq1c}
\dot F
& = &
\nu\nu_E E - \delta_F F.
\end{eqnarray}
\end{subequations}

The life of {\em Aedes} mosquitoes goes through different phases: eggs, larvae, pupae, and finally adults capable of reproducing.
The last phase is aerial, but the three first ones are all aquatic, and subject to competition for space and food resources in the breeding sites.
The state variable of the model has three components, $E,M$ and $F$.
 The component $E$ merges the whole aquatic phase in a single quantity, and $M$ represents the (fertile) males.
In natural conditions, the females are inseminated immediately after birth and produce offspring.
In presence of sterile males (see the controlled system~\eqref{eq2}), their mating with the latter do not produce viable eggs.
In the present model, the variable $F$ represents only those (fertile) females {\em that are inseminated by a fertile male}.
The recruitment term in~\eqref{eq1a} models egg breeding and incorporates `skip oviposition' behavior: the females deposit comparatively fewer eggs in sites with already high occupation rate.
Notice that this term doesn't display the male population $M$.
This apparently paradoxical modelling option is motivated by the fact already mentioned that female insemination occurs  very quickly after hatching.
See however in~\Cref{se72} how to consider extreme male scarcity.

Coherently with these explanations, the parameter $\beta_E$ represents the mean number of eggs laid by a female mosquito per time unit (typically per day); $\nu_E$ the hatching parameter (integrating survivorship and development during the whole aquatic phases); $\nu$ the sex-ratio in offspring; $K$ the egg carrying capacity; and $\delta_E, \delta_M, \delta_F$ the death rates of each of the three respective categories.



\paragraph*{$\bullet$ Controlled system}
\begin{subequations}
\label{eq2}
\begin{eqnarray}
\label{eq2a}
\dot E
& = &
\beta_E F \left(
1 - \frac{E}{K}
\right) - (\nu_E+\delta_E) E,\\
\label{eq2b}
\dot M
& = &
(1-\nu)\nu_E E - \delta_M M,\\
\label{eq2c}
\dot F
& = &
\nu\nu_E E \frac{M}{M+\gamma M_s}- \delta_F F,\\
\label{eq2d}
\dot M_s
& = &
u - \delta_s M_s.
\end{eqnarray}
\end{subequations}

The evolution of the compartments $E$ and $M$ obey the same equations in the controlled system.
The latter contains the supplementary equation~\eqref{eq2d}, which accounts for the evolution of the population of sterilized males $M_s$.
The latter, whose mortality rate is denoted $\delta_s$, are released as adults in the field with an instantaneous rate $u\geq 0$:
this time function is the control variable of the problem.
In the present paper, as well as in~\cite{Bidi:2025aa,Cristofaro:2024aa,Bidi:2023ab} we consider that this variable varies continuously in time.
(In fact, the releases are achieved periodically, so it is more realistic to consider impulsive, periodic, releases~\cite{Bliman:2019aa}.)
The main subtlety of system~\eqref{eq2} lies in~\eqref{eq2c}: here, due to the presence of the sterilized males, only a reduced proportion of mating gives rise to viable eggs.
The parameter $\gamma$, usually smaller than $1$, represents the relative mating competitiveness of the sterile males.
When $\gamma=1$, the fraction $\frac{M}{M+M_s}$ is exactly the proportion of fertile males in the total male population.

All parameters in models~\eqref{eq1} and~\eqref{eq2} are assumed positive, with $\nu < 1$.
For Aedes mosquitoes, measurements have shown that the mortality rates usually fulfil:
$\delta_F \leq \delta_M \leq \delta_s$.
While this is quite important in practice, the mathematical results below are valid without these assumptions.

\comment{
For simplicity, one considers control functions $u\in L^\infty_\loc(0,+\infty)$ such that $u \geq 0$.
We call {\em solution of~\cref{eq2}} any quadruple $(E,M,F,M_s)$ of functions in $C^1\times C^1\times C^1\times W^{1,\infty}_\loc$ fulfilling the equations almost everywhere.
For sake of space, proof of the following result is omitted.
}

It is easy to show that for any nonnegative initial condition and nonnegative control input, the solution of~\eqref{eq2} takes on nonnegative values, and $(E(t),M(t),F(t))$ converges to the forward-invariant order interval $\llbracket \bfE_0, \bfE^U \rrbracket$, where $\bfE_0 := (0,0,0)$, $\bfE^U := K\left(
1, \frac{(1-\nu)\nu_E}{\delta_M},\frac{\nu\nu_E}{\delta_F}
\right)$.
Moreover, $M(0)>0$ implies $M(t)>0$, $t> 0$; and $E(0)>0$ or $F(0)>0$ implies $E(t),M(t),F(t)>0$, $t> 0$.

\section{CONTROL PRINCIPLE}
\label{se2}

We provide here the basic idea used for control synthesis.
This requires precise analysis of the existence and stability of the equilibrium points of~\eqref{eq1}.
This is first recalled in \Cref{se21}, after which the key result is displayed in \Cref{se22} (\Cref{th2}).

\subsection{Basic offspring number and analysis of the uncontrolled system~\eqref{eq1}}
\label{se21}

We recall in this section the asymptotic behaviour of the uncontrolled system~\eqref{eq1}, highlighting the key role of the {\em basic offspring number}.
Related analysis results and demonstrations have been provided in~\cite{Anguelov:2012aa,Anguelov:2020aa,Bidi:2025aa}.

Ordering the state vector components as $(E,M,F)$, the Jacobian matrix of system~\eqref{eq1} is
\begin{equation*}
\begin{pmatrix}
- \frac{\beta_E}{K}F - (\nu_E+\delta_E) & 0 & \beta_E \left(
1 - \frac{E}{K}
\right)\\
(1-\nu)\nu_E & - \delta_M & 0\\
\nu\nu_E & 0 & - \delta_F
\end{pmatrix}.
\end{equation*}
This matrix is a Metzler matrix at any $(E,M,F)$ in the forward-invariant set $\llbracket \bfE_0, \bfE^U \rrbracket$,
so that system~\eqref{eq1} is cooperative~\cite{Smith:1995aa} in $\llbracket \bfE_0, \bfE^U \rrbracket$, see~\cite[Proposition 3.1.1]{Smith:1995aa}.


The system of equations defining the equilibrium points of system~\eqref{eq1} is
$0 = \beta_E F \left(
1 - \frac{E}{K}
\right) - (\nu_E+\delta_E) E
= (1-\nu)\nu_E E - \delta_M M
= \nu\nu_E E - \delta_F F$,
and the {\em extinction equilibrium} $\bfE_0$ 
is always a solution.

On the other hand, due to the last two equations, all components of an equilibrium have to respect given relative proportions, so that any possible positive equilibrium is indeed {\em strictly positive}, due to~\eqref{eq1b} and~\eqref{eq1c}.
Moreover, eliminating the nonzero value $F$ in the first equation thanks to the third one yields
$\beta_E \frac{\nu\nu_E}{\delta_F} \left(
1 - \frac{E}{K}
\right) = \nu_E+\delta_E$,
which has zero or one positive solution, according to the sign of $\beta_E \nu\nu_E - (\nu_E+\delta_E)\delta_F$.
The asymptotic behaviour is summarized by the following result.

\begin{theorem}
\label{th1}
Let
$\cN :=
\frac{\beta_E \nu\nu_E}{(\nu_E+\delta_E)\delta_F}$.
If $\cN \leq 1$, then $\bfE_0$ is the unique equilibrium point of system~\eqref{eq1}, and it is GAS in $\Rset_+^3$.
If $\cN > 1$, then, apart from $\bfE_0$, system~\eqref{eq1} also admits the positive equilibrium point
\begin{equation}
\label{eq6}
\bfE^*
:= K\left(
1 - \frac{1}{\cN}
\right) \left(
1, \frac{(1-\nu)\nu_E}{\delta_M}, \frac{\nu\nu_E}{\delta_F}
\right),
\end{equation}
and the latter is GAS in $\Rset_+^3\setminus\{(0,M,0)\ :\ M\geq 0\}$.
If $\cN \neq 1$, then the GAS equilibrium points are indeed exponentially stable.
\hfill $\square$
\end{theorem}

The constant $\cN$ is called the {\em basic offspring number}.
It characterizes the viability of the considered population in absence of control.
Notice that $\bfE^* = \left(
1 - \frac{1}{\cN}
\right) \bfE^U$.

\begin{remark}
Notice that the complete system~\eqref{eq1} is not {\em strongly order preserving}~\cite[p.~2]{Smith:1995aa}, otherwise $\bfE^*$ would be GAS in $\Rset_+^3\setminus\{\bfE_0\}$ instead of $\Rset_+^3\setminus\{(0,M,0)\ :\ M\geq 0\}$.
\end{remark}

\subsection{Principle of the method: control of the apparent reproduction number}
\label{se22}

From now on and in the rest of the paper, we assume that the wild population is viable, that is
\[
\cN > 1.
\]

We study in this section a general principle of control synthesis, which provides insights useful
for state or output feedback synthesis.
It is stated in~\Cref{th2}.

From the point of view of the wild population, the presence of the sterile insects modifies solely the birth term in~\eqref{eq1a}, through the multiplication by the ratio $\frac{M(t)}{M(t)+\gamma M_s(t)}$.
This is the idea exploited here for control synthesis.
See also~\cite{Bliman:2019aa,Bidi:2025aa}.
\begin{theorem}
\label{th2}
The time function, called {\em apparent reproduction number}, 
$\cN_\app(t) := \cN \frac{M(t)}{M(t)+\gamma M_s(t)}$ is well-defined along
every trajectory of~\eqref{eq2} distinct from the extinction equilibrium $\bfE_0$.
Assume that for such trajectory, the control $u(\cdot)$ is chosen in such a way that
\begin{equation}
\label{eq9}
\exists\, \theta\in (0,1),\ \exists\, T\geq 0,\ \forall t\geq T,\qquad
\cN_\app(t) \leq \theta.
\end{equation}
Then $(E(t),M(t),F(t))$ converges exponentially to $\bfE_0$, uniformly relatively to its initial value in $\llbracket \bfE_0, \bfE^U \rrbracket$.
\hfill $\square$
\end{theorem}


\paragraph*{Sketch of proof of \Cref{th2}}

When~\eqref{eq9} holds, then $\nu\nu_E \frac{M(t)}{M(t)+\gamma M_s(t)}
= \nu\nu_E \frac{\cN_\app(t)}{\cN}
\leq
\frac{(\nu_E+\delta_E)\delta_F}{\beta_E}\theta$, $t\geq T$.
One may then replace~\eqref{eq2c} by the inequality
$\dot F \leq \frac{(\nu_E+\delta_E)\delta_F}{\beta_E}\theta E - \delta_F F$.
A comparison argument then yields the convergence.
\hfill {\small $\blacksquare$}

\comment{
The inequalities~\eqref{eq10a},~\eqref{eq10b},~\eqref{eq10c} are decoupled from~\eqref{eq10d}.
The corresponding set of three differential equations obtained by substituting $\leq$ for $=$ is formally identical to system~\eqref{eq1}, with a basic offspring number
$\cN' = \frac{\beta_E}{(\nu_E+\delta_E)\delta_F} \frac{(\nu_E+\delta_E)\delta_F}{\beta_E}\theta = \theta < 1$.
The extinction equilibrium of this system is thus GAS.
As it is a cooperative system, it may serve as a comparison system.
One deduces that each of its solutions dominates the solution of~\eqref{eq2} departing from the same initial condition.
As a conclusion, the solution of~\eqref{eq2} also converges to 0.

Let $\varepsilon >0$.
Adding $\varepsilon E$, $\varepsilon M$, $\varepsilon F$ respectively in the right-hand sides of~\cref{eq10a,eq10b,eq10c}, amounts to subtracting $\varepsilon$ from $\delta_E, \delta_M, \delta_F$ and does not modify the stability of the system for small enough $\varepsilon>0$, as the corresponding value of $\cN'$ is continuous with respect to $\varepsilon$.
This shows that the convergence is exponential, and achieves the proof of \Cref{th2}. \hfill {\small $\blacksquare$}
}


\begin{remark}
The speed of convergence of the population to zero under a control that verifies~\eqref{eq9} may be guaranteed (that is, estimated from below) using 
a comparison system.
On the other hand, whatever the control $u$,
$\dot E \geq - (\nu_E+\delta_E) E$,
$\dot M \geq - \delta_M M$ and
$\dot F \geq - \delta_F F$.
This imposes hard limitations on the convergence speed, usually not restrictive in practice.
\end{remark}

Ensuring condition~\eqref{eq9} is thus sufficient to stabilize the origin of the wild insect equations.
We will apply this idea in~\Cref{se3} (resp.~\Cref{se4}) to the synthesis of stabilizing state feedback (resp.~output feedback) control laws through continuous-time releases.
It is intuitively reasonable to expect that~\eqref{eq9} may be fulfilled provided the released quantities are `large enough'.
However, we also aim at constructing control laws that vanish when reaching the control goal.
We will see in the sequel that it is possible to complete stabilization with a {\em finite} total amount of released insects, i.e.~with the property
\begin{equation}
\label{eq24}
\int_0^{+\infty} u(t)\cdot dt < +\infty.
\end{equation}

Before going further, notice that~\eqref{eq9} is equivalent to
$\exists\, \theta\in (0,1)$, $\exists\, T\geq 0$, $\forall t\geq T$,
$\frac{M_s(t)}{M(t)} \geq \alpha(\theta)
:= \frac{1}{\gamma}\left(
\frac{\cN}{\theta} - 1
\right)$.
The map $\theta\mapsto \alpha(\theta)$ is {\em decreasing} and maps $(0,1)$ into $(\frac{1}{\gamma}(\cN - 1), +\infty)$, and finally~\eqref{eq9} is equivalent to
\begin{equation}
\label{eq14}
\exists\, \alpha > \alpha_\crit := \frac{1}{\gamma} (\cN - 1),\ \exists\, T\geq 0,\ \forall t\geq T,\
\frac{M_s(t)}{M(t)} \geq \alpha.
\end{equation}
This form is used below as cornerstone for control synthesis.

%

\section{STATE FEEDBACK SYNTHESIS}
\label{se3}

In this section we apply the idea developed in~\Cref{se2} to obtain state feedback controls.
We aim to choose $u$ in such a way as to achieve condition~\eqref{eq14} (or equivalently~\eqref{eq9}) by use of permanent releases, and then benefit from the application of \Cref{th2}.
For this, we build on an idea introduced in~\cite{Bliman:2021aa}, improved and exploited in the sequel.

Let $\alpha_\crit < \alpha < \alpha'$ and $\omega >0$.
If $u$ is chosen in such a way that
\begin{equation}
\label{eq15}
\frac{d(M_s - \alpha' M)}{dt} + \omega (M_s(t)- \alpha'M(t)) \geq 0,\qquad
t\geq 0,
\end{equation}
that is $\dot M_s \geq \alpha' \dot M - \omega (M_s(t)- \alpha'M(t))$, then by integration
\begin{equation}
\label{eq18}
M_s(t) - (M_s(0) - \alpha' M(0))e^{-\omega t}
\geq  \alpha' M(t)
>  \alpha M(t),
\end{equation}
for any $t\geq 0$, and thus
$\displaystyle \liminf_{t\to +\infty}\ \textstyle (M_s(t) - \alpha M(t)) \geq 0$.
However, this does {\em not} imply that $M_s(t) \geq \alpha M(t)$ on $[T,+\infty)$ for large enough $T>0$.
In fact,~\eqref{eq18} implies
$\frac{M_s(t)}{M(t)}
\geq  \alpha' + \frac{\alpha' M(0) - M_s(0)}{M(t)} e^{-\omega t}$.
But the right-hand side here is {\em not} guaranteed to be positive as $M(t)$ goes to zero: the numerator of the fraction could possibly be negative and the fraction itself could take negative arbitrarily large values as $M(t)\to 0$.

Alternatively to~\eqref{eq15}, consider the condition
\begin{equation}
\label{eq20}
\frac{d}{dt}\left(
\frac{M_s(t)}{M(t)} -\alpha'
\right) + \omega \left(
\frac{M_s(t)}{M(t)} -\alpha'
\right) \geq 0,\quad
t \geq 0.
\end{equation}
When~\eqref{eq20} is verified, one obtains
$\frac{M_s(t)}{M(t)}
\geq \alpha' + \left(
\frac{M_s(0)}{M(0)} -\alpha'
\right) e^{-\omega t}$, $t \geq 0$,
which ensures the desired property~\eqref{eq14}, contrary to~\eqref{eq18}.
Inequality~\eqref{eq20} is expressed equivalently as
$\dot M_s \geq \frac{M_s(t)}{M(t)} \dot M(t) + \omega ( \alpha' M(t) - M_s(t))$, $t\geq 0$,
This imposes the following unilateral constraint on the value of $u$:
\begin{equation}
\label{eq22}
u (t) \geq \frac{M_s(t)}{M(t)} \dot M(t) + \omega \alpha' M(t) + (\delta_s- \omega) M_s(t),\quad t \geq 0,
\end{equation}
while, for its part, \eqref{eq15} yields
\begin{equation}
\label{eq17}
u(t) \geq \alpha' \dot M(t) + \omega \alpha' M(t) + (\delta_s - \omega) M_s(t),
\hspace{.4cm} t \geq 0.
\end{equation}
This is not fully satisfying, as the first right-hand side term of~\eqref{eq22} could possibly take on large or unbounded values when the denominator vanishes, requiring large or unbounded values for the control $u$.

To circumvent this difficulty, we slightly adapt condition~\eqref{eq22} and~\eqref{eq17} in a way convenient for our purpose.


\begin{theorem}
\label{th3}
Assume that for some given constants $\alpha > \frac{1}{\gamma}(\cN - 1)$, $\omega >0$, one has
\begin{equation}
\label{eq27}
u(t) \geq \alpha |\dot M|_+ + \omega \alpha M(t) + (\delta_s - \omega) M_s(t),\qquad t \geq 0.
\end{equation}
Then, for all trajectories of system~\eqref{eq2}, $(E(t),M(t),F(t))$ converges exponentially to the extinction equilibrium $\bfE_0$.

Moreover, if $\omega \geq \delta_s$ and there exists $C\geq 1$ such that
\begin{equation}
\label{eq30}
u(t) \leq C\left|
\alpha |\dot M|_+ + \omega\alpha M(t) + (\delta_s - \omega ) M_s(t)
\right|_+,\quad t \geq 0,
\end{equation}
then~\eqref{eq24} holds.
The same holds if $\omega <\delta_s$ and
\begin{equation}
\label{eq230}
C < \frac{\delta_s}{\delta_s-\omega}.
\end{equation}
\hfill $\square$
\end{theorem}

One has $\alpha |\dot M|_+ + \omega \alpha M = \max\{ \alpha(1-\nu)\nu_E E + \alpha(\omega - \delta_M) M; \omega \alpha M  \}$.
As the release rate $u$ takes on only nonnegative values, one may express~\eqref{eq27} equivalently as
\begin{multline}
\label{eq277}
u(t) \geq \max\{
\alpha(1-\nu)\nu_E E + \alpha(\omega - \delta_M) M + (\delta_s - \omega) M_s;\\
\omega \alpha M + (\delta_s - \omega) M_s;0\},\qquad t \geq 0,
\end{multline}
and~\eqref{eq30} may be transformed similarly.
This form highlights the fact that the expression on the right-hand side of~\eqref{eq277} is {\em piecewise linear} in the state variable.

\Cref{th3} states that any release strategy satisfying condition~\eqref{eq27} from a certain time ensures stabilization of the equilibrium $\bfE_0$.
This condition depends upon the two design parameters $\alpha$ (which defines the aimed `safety value' of $\frac{M_s}{M}$) and $\omega$ (which assesses the convergence speed towards this value).
Thanks to the use of monotone system theory, it is expressed as an {\em inequality} defining minimal release rate, and possesses inherently robustness properties (see also~\Cref{se7}).

It is worth noting that, {\em for any $\omega>0$}, applying~\eqref{eq27} with an equality ensures stabilization with finite total release (i.e.~\eqref{eq24}), as this amounts to ensure~\eqref{eq30} with $C=1$.




\begin{remark}
\label{re2}
Obviously, condition~\eqref{eq27} is more stringent than condition~\eqref{eq17}.
On the other hand, it is also stronger than condition~\eqref{eq22} when $\dot M \leq 0$; or when $\dot M > 0$ and $\alpha
- \frac{M_s(t)}{M(t)}$ is still positive.
\end{remark}

\paragraph*{Sketch of proof of \Cref{th3}}

One first deduces from~\eqref{eq27} that $z(t) := \left|
\alpha
- \frac{M_s(t)}{M(t)}
\right|_+$ fulfils $\dot z \leq -\omega z$ when $z>0$, so that $z(t)$ vanishes when $t\to +\infty$.
This yields convergence to $\bfE_0$.
When~\eqref{eq30} holds, then $\dot M_s + \min \{C \omega - (C-1) \delta_s; C \delta_s \}\, M_s \leq C\alpha (|\dot M|_+ + \omega M)$, where $(C-1) \delta_s- C \omega <0$ by assumption and $|\dot M|_+ + \omega M$ vanishes exponentially.
Gronwall's lemma then shows exponential convergence of $M_s$ and $u$ to zero.
\hfill {\small $\blacksquare$}

\section{ROBUSTNESS ISSUES}
\label{se7}

We discuss in this section how the previous stabilization result may be extended in presence of uncertainty.
This flexibility comes as a major benefit of the monotonicity-based argument.
We present in~\Cref{se71} a result concerning uncertainty on the parameters of system~\eqref{eq2}; and in~\Cref{se72} an important extension, in case where an Allee effect is added to the model (dynamic uncertainty).

\subsection{Parametric uncertainties}
\label{se71}

Due to~\Cref{th3},~\eqref{eq27} has two remarkable properties.
It aims at ensuring permanently a value of the ratio $\frac{M_s}{M}$ sufficient to reduce the wild population until it gets extinct.
As such, it depends only upon the dynamics~\eqref{eq2d} of $M_s$ and upon the value of $M$ and $\dot M$.
It takes the form~\eqref{eq277} when~\eqref{eq2b} is used, but other dynamics of $M$ could be handled in the same way.
Second, condition~\eqref{eq27} does not impose a given supply rate, but only a minimal value for the latter.
Both this inequality form and the previously mentioned genericity of the approach ensure intrinsic robustness properties.

We propose here a robustness result as illustration, which ensures the same results than~\Cref{th3}, based on upper and lower ($U$ and $L$) estimates of the parameters and of the state variables.
\begin{theorem}
\label{th5}
Let $\alpha > \frac{1}{\gamma^L}(\cN^U - 1)$, $\omega > 0$, and assume $u$, defined for $t \geq 0$, fulfils
$u \geq \max\{
\alpha(1-\nu^L)\nu_E^U E^U + \alpha(\omega - \delta_M^L) M^U + (\delta_s^U - \omega) M_s^L;
\omega \alpha M^U + (\delta_s^U - \omega) M_s^L;0\}$,
for $\gamma^L\leq\gamma$, $\cN \leq\cN^U$, $\nu^L \leq \nu$, $\nu_E^U \geq \nu_E$, $\delta_M^L \leq \delta_M$, $\delta_s \leq \delta_s^U$,
and for signals $E^U(t) \geq E(t)$, $M^U(t) \geq M(t)$, $M_s^L(t) \leq M_s(t)$ for any $t\geq 0$.
Then, provided
$\omega \geq \max\{ \delta_M^L, \delta_s^U\}$,
$(E(t),M(t),F(t))$ converges exponentially to the extinction equilibrium $\bfE_0$ for every trajectory of system~\eqref{eq2}.
\hfill $\square$
\end{theorem}

\Cref{th5} shows that the proposed feedback rule is robust against parametric uncertainty, but also against estimation error on the state variables.

\comment{
\paragraph*{Proof of~\Cref{th5}}
Under condition~\eqref{eq53}, one has
$(\omega - \delta_M^L) M^U
- (\omega - \delta_M) M
= (\omega - \delta_M^L) (M^U-M) +(\delta_M - \delta_M^L)M \geq 0$ and
$(\delta_s^L - \omega) M_s^L - (\delta_s - \omega) M_s
= (\delta_s^L - \omega) (M_s^L - M_s) + ( \delta_s^L - \delta_s)M_s \geq 0$.
Under these conditions,~\eqref{eq278} implies~\eqref{eq277}, which is just inequality~\eqref{eq27}.
The proof is achieved by use of~\Cref{th3}, using the fact that $\alpha > \frac{1}{\gamma^L}(\cN^U - 1) \geq \frac{1}{\gamma}(\cN - 1)$.
\hfill {\small $\blacksquare$}
}

\subsection{Dynamic uncertainties: Allee effect}
\label{se72}

A variant of system~\eqref{eq2} is considered in~\cite{Bossin:2019aa}, with~\eqref{eq2c} replaced by
\addtocounter{equation}{-14}
\begin{subequations}
\addtocounter{equation}{4}
\begin{equation}
\label{eq2e}
\dot F = \nu\nu_E E \left(
1-e^{-\beta(M+\gamma M_s)}
\right) \frac{M}{M+\gamma M_s}- \delta_F F,
\end{equation}
\end{subequations}
for some $\beta>0$.
This term introduces an {\em Allee effect}~\cite{Fauvergue:2013aa}, that is a positive correlation between population density and growth rate.
\addtocounter{equation}{13}
The following result handles such a case.
\begin{theorem}
\label{th6}
Assume that for given constants $\alpha > \frac{1}{\gamma}(\cN - 1)$, $\omega >0$, $u$ fulfils~\eqref{eq27}.
Then, for all trajectories of system~\eqref{eq2a}-\eqref{eq2b}-\eqref{eq2e}-\eqref{eq2d}, $(E(t),M(t),F(t))$ converges exponentially to the extinction equilibrium $\bfE_0$.
\hfill $\square$
\end{theorem}

\comment{
\begin{proof}
First,~\eqref{eq27} implies that~\eqref{eq14} holds.
Putting $\theta = \frac{\cN}{1+\gamma\alpha} < 1$, one thus gets asymptotically
$\dot F \leq \frac{\nu\nu_E}{1+\gamma\alpha} E - \delta_F F
= \frac{(\nu_E+\delta_E)\delta_F}{\beta_E}\theta E - \delta_F F$, for any $t\geq T$.
Using comparison system as in the proof of~\Cref{th2}, one then deduces the exponential convergence towards $\bfE_0$.
\end{proof}
}

\section{OUTPUT FEEDBACK SYNTHESIS}
\label{se4}

In this section, we first review in~\Cref{se41} the quantities that may be available for measurement.
Then we provide in~\Cref{se42} a stabilization result as an interval-observer based output-feedback control.

\subsection{Measured quantities}
\label{se41}

Measurements are produced by using a variety of traps, either aimed at the aquatic phases (as `ovitraps') or the adult one, possibly boosted by mark-release-recapture method.
In absence of prior marking of the sterile males, it is quite expensive to distinguish between the sterile and the fertile males.
It is even more complex to determine whether a given female has been inseminated by a sterile or a fertile male.
Therefore, one simply assumes here, as in~\cite{Bidi:2025aa}, that are measured two global quantities, namely the total number of male insects $M_\tot(t)$, and the total number of female insects $F_\tot(t)$.
(Notice that there also exist counting methods for mosquito eggs deposited on ovitrap sticks.
The compartment $E$ models here all the aquatic phases, so the use of such data would require us to revisit the model.
This topic is not delved into in depth for lack of space.)

\begin{itemize}
\item
The male insects either come from the wild population, or are sterile males still alive after their previous introduction in the field.
Thus, $M_\tot(t) = M(t) + M_s(t)$.
When both equation parameters and released amounts are perfectly known, using an observer allows to reconstruct exactly the value of $M_s$ by using equation~\eqref{eq2d}.
In other terms, measuring $M_\tot$ allows us to get access to $M$.

\item
On the other hand, it is reasonable to assume that every female insect is fertilized once immediately after hatching, and that the adult phase takes place identically depending on whether the female has been fertilized by a fertile male or a sterile male.
Therefore, by  analogy with~\eqref{eq2c}, the number $F_s := F_\tot-F$ of females fertilized by a sterile evolves according to
\[
\dot F_s = \nu\nu_E E \frac{\gamma M_s}{M+\gamma M_s}- \delta_F F_s.
\]
Consequently, the total number $F_\tot = F_s +F$ of female insects present in the field fulfils equation~\eqref{eq1c}.
This definition is quite different from the static model $F_s = \frac{M_s}{M}F$ considered in~\cite{Bidi:2023ab}, which amounts to taking $F_\tot :=  \frac{M+M_s}{M}F$ (when $\gamma=1$).
In any case, adding the measurement $F_\tot$ implies addition of a new state component, decoupled from the rest.
For this reason, this quantity is not used in the sequel.
\end{itemize}

\subsection{Output-based stabilizing control laws}
\label{se42}

We show here how to use interval observers~\cite{Efimov:2016aa} in complement to~\Cref{th3} to obtain upper estimates of the state variables when the latter is not completely measured.
We assume that $M_\tot = M+M_s$ is measured.
Monotonicity calls for attempting to treat the case of other outputs in the same vein.

Define $x:= \begin{pmatrix} E & M & F \end{pmatrix}\t$, and $f$ the function such that~\eqref{eq2} is represented as

%

\begin{subequations}
\label{eq65}
\begin{equation}
\label{eq65a}
\begin{pmatrix} \dot x \\ \dot M_s \end{pmatrix} = f(x,M_s,u).
\end{equation}
In the definition of $f$, one changes
$\left(
1 - \frac{E^U}{K}
\right)$ to $\left|
1 - \frac{E^U}{K}
\right|_+$.
This does not modify the trajectories of the corresponding solution of~\eqref{eq2} as the latter lies inside the forward-invariant interval $\llbracket \bfE_0, \bfE^U \rrbracket$.

Define also the estimator components $x^U, M_s^L$ by
\begin{multline}
\label{eq65b}
\begin{pmatrix} \dot x^U \\ \dot M_s^L \end{pmatrix}
= f(x^U,M_s^L,u)\\
+  \begin{pmatrix} 0 & \omega_M & 0 & 0 \end{pmatrix}\t (M_\tot - M_s^L-M^U).
\end{multline}
\end{subequations}
The next result shows how to stabilize the equilibrium $\bfE_0$ of system~\eqref{eq65a} by dynamic output feedback computed from the interval observer~\eqref{eq65b}.

\begin{theorem}
\label{th7}
Assume that for some $\alpha > \alpha_\crit$, $\omega >0$, $\omega_M \geq 0$, one has
\begin{equation}
\label{eq280}
u(t) \geq \alpha |\dot M^U|_+ + \omega\alpha M^U(t) + (\delta_s-\omega) M_s^L(t),
\hspace{.24cm}
t \geq 0.
\end{equation}
Then, for any solution of system~\eqref{eq65} such that
\begin{equation}
\label{eq67}
x(t) \leq x^U(t),\qquad
M_s(t) \geq M_s^L(t) \geq 0
\end{equation}
for $t=0$,~\eqref{eq67} holds for any $t\geq 0$, and $(E(t),M(t),F(t))$ converges exponentially to the extinction equilibrium $\bfE_0$.

Moreover, if $\omega \geq \delta_s$ and there exists $C\geq 1$ such that
$u(t) \leq C\left|
\alpha |\dot M^U|_+ + \omega\alpha M^U(t) + (\delta_s - \omega ) M_s^L(t)
\right|_+$
for any $t \geq 0$, then~\eqref{eq24} holds.
The same holds if $\omega <\delta_s$ and~\eqref{eq230} is fulfilled.
\hfill $\square$
%
\end{theorem}

Initialization of the observer is achieved according to the uncertainty range.
The definition of the invariant set $\llbracket \bfE_0, \bfE^U \rrbracket$ 
or the value of $\bfE^*$ in~\eqref{eq6} (as the order interval $\llbracket\bfE_0;\bfE^*\rrbracket$ is also forward-invariant) may serve as basis.

 
Here, condition~\eqref{eq280} reads (compare with~\eqref{eq277})
$u(t) \geq \max\{
\alpha(1-\nu)\nu_E E^U + \alpha(\omega - \delta_M - \omega_M) M^U + \alpha\omega_M M_\tot + (\delta_s - \omega - \alpha\omega_M) M_s^L;
\alpha \omega M^U + (\delta_s - \omega) M_s^L;0
\}$ for any $t \geq 0$.
 
 Notice that the stabilization result of~\Cref{th7} also holds when $\omega_M=0$, as in this case both the state $x$ and the state estimate $x^U$ converge to zero.
 Taking positive values of $\omega_M$ should yield faster estimation.
 
 As a last remark, one may show that when only $M$ or only $M_s$ is measured, the same argument implies that the conclusions of \Cref{th7} remain valid when replacing $(M_\tot - M_s^L-M^U)$ by $(M - M^U)$ or $(M_s - M_s^L)$ in the feedback term.

\section{CONCLUSION}
\label{se5}

We studied in this paper the implementation by feedback control of the Sterile Insect Technique against {\em Aedes} mosquitoes.
State-feedback and output-feedback control laws have been proposed, and the
handling of parametric and dynamical uncertainties of the model has been considered.
The approach adopted here appears as simpler, more flexible and more powerful than the ones previously considered. 
For sake of space, complete proofs of the results and numerical illustrations were deferred to a subsequent extended publication, which will also include new contributions.

Biological control methods aim 
at controlling specific pest or vector using another living organism.
They raise interesting questions to control theory.
In this respect, two original technical points examplified in this note are worth noticing.
First, a central role was played by the basic offspring number.
This quantity shrinks when the proportion of sterile males in the total male population increases, so that reducing and keeping it below the threshold level under which the population is unviable, may be used as a control goal to guarantee elimination.
On the other hand, the results exposed here were all obtained thanks to monotone system theory, and the other directions  explored in the forthcoming complete version of this work are based on the same foundation.
Monotone system theory appears quite well-fitted, powerful and flexible to study the interaction of competitive species in such situations.
We believe the two previous features constitute interesting leads for the design of control strategies of proven effectiveness.



\section*{ACKNOWLEDGMENT}

The author expresses grateful thanks to Professor Amit Bhaya (COPPE/UFRJ, Brazil) for careful reading and valuable comments on the manuscript.

\bibliographystyle{plain}
\bibliography{Biblio-SIT}

\end{document}